\documentclass[lettersize,journal]{IEEEtran}

\usepackage{graphicx} 

\usepackage{algorithm}
\usepackage{algorithmic}
\usepackage{amsmath}
\usepackage{amssymb}
\usepackage{amsthm}
\usepackage{cleveref}
\usepackage{xcolor}
\usepackage{caption}
\usepackage{subcaption}

\usepackage{macros}

\title{Image Priors for Single-Shot Coded Diffraction Phase Retrieval}

\author{\IEEEauthorblockN{Zhiyuan Hu\IEEEauthorrefmark{1}, YoonSeok Baek\IEEEauthorrefmark{2}, Sylvain Gigan\IEEEauthorrefmark{2}, Michael Unser\IEEEauthorrefmark{1}, Jonathan Dong\IEEEauthorrefmark{1}} \\
\IEEEauthorblockA{\IEEEauthorrefmark{1}Biomedical Imaging Group, \'Ecole polytechnique f\'ed\'erale de Lausanne, Lausanne, Switzerland} \\ 
\IEEEauthorblockA{\IEEEauthorrefmark{2}Laboratoire Kastler Brossel, ENS-Universit\'e PSL, CNRS, Sorbonne Universit\'e, Coll\`ege de France, Paris, France}
}

\begin{document}

\maketitle

\begin{abstract}
Phase retrieval recovers complex-valued objects from intensity-only measurements.
In its classical Fourier form, reconstruction suffers from twin-image stagnation.
Coded diffraction imaging, which adds a random phase mask before the Fourier transform, has so far required several masks.
We show that an image prior makes single-shot coded diffraction imaging possible, with a single random mask followed by a Fourier transform.
We validate this claim in simulation and in an optical experiment, in which a spatial light modulator displays the phase object together with the random mask.
Our experimental validation covers five test signals and 530 images from the public Brno quantitative phase microscopy dataset of five adherent cell lines.
With no image prior or with Fourier-only measurements, reconstruction of complex-valued objects from a single diffraction pattern consistently fails.
An out-of-the-box pretrained image denoiser enables accurate reconstruction, in simulation even with fewer measurements than unknowns.
Five regularization frameworks built on the same denoiser and one learned weakly convex regularizer all reconstruct the experimental images with peak signal-to-noise ratios within one decibel of each other. 
A single random phase mask and image priors thus turn coded diffraction imaging into a practical single-shot modality for quantitative phase imaging.
\end{abstract}


\section{Introduction}
Phase retrieval refers to reconstructing a complex-valued signal from intensity measurements obtained after a linear transformation. 
In imaging, this corresponds to recovering the object-plane field from detector-plane intensity measurements. 
It arises in a wide range of domains, including crystallography~\cite{harrison1993phase}, astronomy~\cite{dainty1987phase}, microscopy~\cite{gerchberg1972practical} and optical metrology~\cite{fienup1993hubble}.
Phase is often the quantity of interest: weakly absorbing specimens such as unstained cells are almost transparent in intensity images, whereas their phase delays reveal high-contrast structures~\cite{park2018quantitative,marquet2005digital}.
Moreover, phase imaging requires no exogenous labels and is thus free of photobleaching and largely free of phototoxicity, which permits long-term observation of living samples~\cite{park2018quantitative,zangle2014live}.
These benefits come at the cost of an inverse problem to solve, whose difficulty depends on the linear operator at hand.


The structure of the linear operator, which maps the object to the detector plane, is set by the imaging modality.
The first and extensively studied case is Fourier phase retrieval (FPR), also referred to as coherent diffraction imaging, in which the operator is a two-dimensional Fourier transform.
A rich family of algorithms has been developed for this setup, including the Gerchberg-Saxton algorithm, error reduction, hybrid input-output, and relaxed averaged alternating reflections~\cite{gerchberg1972practical,fienup1982phase,luke2004relaxed}. Many of them are framed using alternating projections between the measurement set and a constraint set.
In practice, however, Fourier phase retrieval remains ill-posed and requires strong prior constraints, such as a known support and nonnegativity, which preclude recovery of arbitrary complex-valued objects. 
Moreover, its algorithms are subject to the well-known twin-image stagnation problem~\cite{fienup1986phase}.
To facilitate reconstruction with more measurements, a common scheme is to shift the position or the angle of the illumination between exposures, leading to ptychography and Fourier ptychography, respectively~\cite{rodenburg2004phase,zheng2013wide}. Another approach places random phase masks just after the object and records the resulting coded diffraction patterns (CDPs)~\cite{candes2015phasea}.
This model is a standard setting for phase retrieval with theoretical guarantees, with which convex relaxations and gradient-based solvers recover the signal exactly, up to a global phase~\cite{candes2015phasea,candes2015phase,gross2017improved}.
Existing works on coded diffraction imaging, however, bear two limitations.
First, every existing guarantee requires multiple masks, at least three in practice~\cite{seaberg2015coherent}, corresponding to multiple exposures of the same object.
Second, experimental realizations of coded diffraction imaging remain scarce. They either acquire several patterns with different masks~\cite{seaberg2015coherent,katkovnik2017computational,wu2019wish} or use a Fresnel operator instead of a Fourier transform~\cite{horisaki2016single,gao2026singleshot}.

Phase retrieval with a dense random sensing matrix, whose entries are i.i.d. Gaussian, is the standard configuration for theoretical studies. Exact thresholds for recovery have been derived and validated~\cite{mondelli2019fundamental,luo2019optimal,maillard2020phase}. However, such matrices, which correspond to scattering in a complex medium, bear a prohibitive quadratic computational complexity in the number of pixels.
Structured random phase retrieval (SRPR)~\cite{hu2026structured} has been recently proposed to bridge the gap between real-world Fourier-based modalities and i.i.d. random ones.
Its forward model is a cascade of Fourier transforms and random diagonal matrices, and can directly be implemented using lenses and diffusers.
With two or more layers, the model enables the same robust reconstruction quality as using a dense random sensing matrix, with computational cost reduced from quadratic to log-linear in the number of pixels~\cite{hu2026structured}.
The one-layer case is exactly phase retrieval from a single CDP.
Without an image prior, its reconstruction by spectral methods and gradient-based solvers fails regardless of the number of measurements~\cite{hu2026structured}.

Reducing the number of layers needed for a successful reconstruction is nevertheless of great practical interest, as every reduced layer simplifies the optical hardware. The one-layer system requires only a single mask and a single lens.
We have shown that, in the two-layer model, regularization already enables reconstruction from far fewer measurements~\cite{ducotterd2025breaking}. 
The effect of regularization on the one-layer case, however, has yet to be investigated.

In this paper, we study the effect of regularization on FPR and on phase retrieval from a single CDP, using pretrained image denoisers as prior.
Our contributions are the following:
\begin{itemize}
    \item We show, both numerically and experimentally, that machine-learning regularization makes phase retrieval from a single CDP possible. With a pretrained denoiser as the prior, a single pattern yields an accurate and robust reconstruction of the full complex-valued object, whereas every prior-free algorithm fails at any sampling ratio.
    \item We experimentally demonstrate the approach with an optical system in which a spatial light modulator (SLM) encodes both the object phase and the random mask.
    \item We benchmark six machine-learning regularization frameworks (plug-and-play with half-quadratic splitting or proximal gradient steps, regularization by denoising, a score-based prior, the least squares residual regularizer and a learned weakly convex regularizer) on a large dataset of 530 quantitative phase images from five different cell lines.
    \item We investigate two strategies to regularize a complex-valued signal with a real-valued denoiser, the real/imaginary and the magnitude/phase splitting schemes, and adopt the former.
    \item We show that the same prior fails for Fourier phase retrieval of complex-valued objects.
    Together with the two-layer results of~\cite{ducotterd2025breaking}, this establishes single-shot coded diffraction imaging as the simplest structured random model that benefits from regularization.
\end{itemize}

\section{Background}
\subsection{Forward Model}
Phase retrieval reconstructs a complex-valued signal $\x \in \C^{n}$ from intensity measurements $\y \in \R^{m}$ of the form
\begin{equation} \label{eq:pr-forward}
    \y = |\A\x|^{2} + \n,
\end{equation}
where the sensing matrix $\A \in \C^{m \times n}$ models the measurement physics, and $\n \in \R^{m}$ is the measurement noise; $|\cdot|$ denotes the element-wise modulus.
The sampling ratio $\alpha = m/n$ is the number of measurements per unknown, and sets the difficulty of reconstruction. Reconstructing the signal $\x$ is always up to a global phase shift $\phi \in [-\pi, \pi)$, as $|\A \x|^2 = |\A(\x \e^{\j \phi})|^2$.

SRPR~\cite{hu2026structured} defines the sensing matrix as
\begin{equation} \label{eq:srpr}
    \A = \mathbf{S} \Bigl[\prod_{i=1}^{N} \F \D_{i} \Bigr]\, \F^{s} \O,
\end{equation}
where
\begin{itemize}
    \item $\F \in \C^{p \times p}$ denotes the discrete Fourier transform (DFT),
    \item each $\D_{i} \in \C^{p \times p}$ is a random diagonal matrix,
    \item $\F^{s}$ with $s \in \{0, 1\}$ is either the DFT or the identity,
    \item $\O \in \R^{p \times n}$ pads the $n$-pixel signal to size $p \geq n$, with zeros unless stated otherwise.
    $\mathbf{S} \in \R^{m \times p}$ optionally keeps $m \leq p$ of the output coefficients, so that $p = \max(m, n)$ as in~\cite{hu2026structured}.
\end{itemize}
Each product $\F\D_{i}$ is one \emph{structured random layer}; the trailing $\F^{s}$ counts as half a layer. We thus define the $\textit{depth}$ of the model as $L = N + s/2$.
In this paper, we mainly focus on the case where $m \geq n$, assuming the undersampling operator $\mathbf{S} = \mathbf{I}$ unless specified otherwise.

When $L = 0.5$, the sensing matrix $\A = \F\O$ corresponds to oversampled FPR.
Besides the irreducible global phase ambiguity, FPR suffers from two other ambiguities: 1. a cyclic translation $\x[\k - \k_0]$ of the signal $\x[\k]$, with $\k$ denoting the pixel index and $\k_0$ the shift bias; 2. the conjugate inversion $\conj{\x}[-\k]$, that is, the complex conjugate of $\x[\k]$ point-reflected through the image center, commonly referred to as the \emph{twin image}.
Any combination of the three ambiguities produces identical measurements.
The uniqueness of the solution up to ambiguities has been proved if the signal is at least two-dimensional and oversampled twice for each dimension~\cite{hayes1982reconstruction}.
In practice, however, the ambiguities generate a landscape of spurious local minima, and reconstruction algorithms stagnate even when uniqueness holds.

On the other hand, CDPs are obtained by modulating the signal with random masks before the Fourier transform, leading to the sensing matrix
\begin{equation} \label{eq:coded-diff}
    \A = \begin{bmatrix} \F\D_{1} \\ \vdots \\ \F\D_{k} \end{bmatrix} \in \C^{kn \times n},
\end{equation}
where each $\D_{i} \in \C^{n \times n}$ is a random diagonal matrix and $k$ is the number of masks.
Convergence guarantees up to a global phase have been established for convex relaxations and gradient-based algorithms when the number of masks $k$ grows polylogarithmically with $n$~\cite{candes2015phasea,candes2015phase,gross2017improved}, and this model serves as a standard benchmark for plug-and-play~\cite{song2022fast} and diffusion-based~\cite{shoushtari2023diffusion} solvers.
The single-shot version, $k = 1$, is exactly the one-layer SRPR model $\A = \F\D\O$.

When $\A$ has i.i.d. complex Gaussian entries, the phase transitions of phase retrieval are known exactly. Weak recovery means an estimate positively correlated with the signal. Without a prior, it is allowed by information theory only for $\alpha > 1$~\cite{mondelli2019fundamental}. 
This is possible in practice for example with spectral methods~\cite{mondelli2019fundamental,luo2019optimal}. Perfect recovery in the noiseless case requires $\alpha > 2$~\cite{maillard2020phase}. Approximate message passing achieves perfect recovery for $\alpha$ above approximately 2.03, slightly above the previous threshold~\cite{maillard2020phase}.
With two or more layers, SRPR empirically matches the performance of dense random matrices at a log-linear cost~\cite{hu2026structured}. Regularization on the two-layer model breaks the weak-recovery threshold~\cite{ducotterd2025breaking}.
In contrast, the one-layer model fails without regularization at any sampling ratio.

\subsection{Reconstruction Algorithms} \label{sec:algorithms}
Modern reconstruction algorithms for phase retrieval are commonly formulated as a variational problem
\begin{equation} \label{eq:variational}
    \min_{\x}\; \fid(\A, \x, \y) + \lambda \reg(\x),
\end{equation}
where the data-fidelity term $\fid$ measures the consistency of the estimate with the measurements, and the regularization term $\reg$ encodes prior knowledge about the signal, weighted by the regularization strength $\lambda \geq 0$. In the Bayesian view, \cref{eq:variational} is maximum a posteriori estimation, with $\fid$ the negative log-likelihood $-\log p(\y | \x)$ and $\lambda \reg$ the negative log-prior $-\log p(\x)$.

Two data-fidelity terms are commonly used. The intensity loss $\frac{1}{2}\||\A\x|^{2} - \y\|^{2}$ is the negative log-likelihood under Gaussian noise on the intensities. The amplitude loss $\frac{1}{2}\||\A\x| - \sqrt{\y}\|^{2}$ approximates the negative log-likelihood under Poisson noise, and it has been found more robust to noise in practice~\cite{yeh2015experimental}.

A basic way to enforce consistency with the measurements is the measurement projection
\begin{equation} \label{eq:gs}
    \x_{t+1} = \A\HH (\sqrt{\y} \odot \e^{\j \angle \A\x_t}),
\end{equation}
which imposes the measured amplitudes and keeps the current phases. Here $(\cdot)\HH$ denotes the conjugate transpose, $\odot$ the elementwise product and $\angle(\cdot)$ the elementwise argument of a complex vector.
The proximal operator of a function $g$ with parameter $\tau > 0$ is $\operatorname{prox}_{\tau g}(\x) = \arg\min_{\mathbf{z}} \tau g(\mathbf{z}) + \frac{1}{2}\|\mathbf{z} - \x\|^{2}$~\cite{parikh2014proximal}.
For a unitary $\A$, \cref{eq:gs} is the projection onto the measurement set $\{\x : |\A\x|^2 = \y \}$, that is, the proximal operator of its indicator function.
For a general $\A$, with the pseudo-inverse in place of $\A\HH$, it is the amplitude projection in the measurement domain followed by a least-squares return to the object domain.
For the zero-padded models of this paper, it is one error-reduction step~\cite{fienup1982phase}.
The Gerchberg-Saxton algorithm is the special case where the sensing matrix is the Fourier transform, and it alternates this projection with the projection onto the known object amplitude~\cite{gerchberg1972practical}.

A relaxed version of \cref{eq:gs} with a step size $\mu \in (0, 1]$ was proposed in~\cite{ducotterd2025breaking},
\begin{equation} \label{eq:relaxed_measurement_projection}
   \x_{t+1} = \x_t + \mu\, \A\HH \bigl(\sqrt{\y} \odot \e^{\j \angle \A\x_t} - \A\x_t\bigr).
\end{equation}
For any $\A$, this is one gradient step of size $\mu$ on the amplitude loss, and $\mu = 1$ recovers \cref{eq:gs} when $\A\HH\A = \mathbf{I}$.
Below a sampling ratio of 1, $\A\HH\A$ is only a projection, and the step with $\mu = 1$ keeps the unmeasured component $(\mathbf{I} - \A\HH\A)\x_t$, which \cref{eq:gs} sets to zero.
For a unitary $\A$, it is also the proximal operator of this loss with parameter $\mu / (1 - \mu)$~\cite{ducotterd2025breaking}.

In recent years, learned regularization terms have greatly improved reconstruction quality in inverse problems. This generally consists in training an image denoiser $D_{\sigma}(\cdot)$, which predicts the clean image $\x$ from its noise-corrupted observation $\x + \n$ at noise level $\sigma$.
A denoiser trained with the mean-squared-error (MSE) loss on clean-noisy image pairs approximates the minimum-mean-squared-error (MMSE) estimator $\mathbb{E}[\x | \x + \n]$~\cite{kadkhodaie2021stochastic}.

\emph{Plug-and-play} (PnP) methods start from the observation that in splitting algorithms such as half-quadratic splitting (HQS) and the alternating direction method of multipliers (ADMM), the regularization term acts only through its proximal operator, which can be treated as a Gaussian denoising step at a noise level set by the splitting parameter~\cite{venkatakrishnan2013plug}.
The motivation is to bypass the explicit design of $\reg$ by replacing its proximal operator with a learned denoiser, leading to the regularization step
\begin{equation} \label{eq:pnp}
    \x_{t+1} = D_{\sigma}(\x_t).
\end{equation}
Combined with a proximal step or with a gradient step on the data-fidelity term, this drop-in replacement yields PnP-HQS~\cite{zhang2021plug} and PnP proximal gradient descent (PnP-PGD)~\cite{ryu2019plug}, respectively.
The ideal MMSE estimator indeed equals the proximal operator of an implicit nonconvex regularization term, which differs from the negative log-prior $-\log p(\x)$~\cite{gribonval2011map}.
The employed trained network, however, is only an approximation of the MMSE estimator and violates the condition of being a nonexpansive gradient field to equal a proximal operator~\cite{gribonval2020characterization}.
Convergence guarantees for generic trained denoisers are thus weaker. If the denoising residual $D_{\sigma} - \mathbf{I}$ is Lipschitz with a sufficiently small constant and the data term is strongly convex, the PnP iteration converges to a unique point, which is however not the minimizer of an explicit objective~\cite{ryu2019plug}.

\emph{Regularization by denoising} (RED) takes the opposite route and keeps the regularization term explicit~\cite{romano2017little}.
It defines the regularization term from the denoiser itself, $\reg(\x) = \frac{1}{2}\inner{\x}{\x - D_{\sigma}(\x)}$.
Under Jacobian symmetry and local homogeneity~\cite{romano2017little}, the gradient of $\reg$ simplifies to the denoising residual $\nabla\reg(\x) = \x - D_{\sigma}(\x)$.
The regularization step is thus a standard gradient descent of form
\begin{equation} \label{eq:red}
    \x_{t+1} = \x_t + \lambda (D_{\sigma}(\x_t) - \x_t),
\end{equation}
with constant $\lambda > 0$. When $\lambda = 1$, it recovers the PnP update. In practice, trained networks do not satisfy these conditions, so $\x - D_{\sigma}(\x)$ is not the gradient of any explicit regularization term, and the RED update is a heuristic iteration whose behavior is set by the denoiser rather than by an objective~\cite{reehorst2019red}.

Gradient-step denoisers are motivated by making the gradient of the learned regularization term exact rather than assumed~\cite{hurault2022gradient}.
Instead of training the network as the denoiser, this approach trains an auxiliary network $N_{\sigma}$ and defines the denoiser as one gradient step on the \emph{least squares residual} (LSR) $\reg(\x) = \frac{1}{2}\|\x - N_{\sigma}(\x)\|^{2}$~\cite{zou2023deep,hertrich2026learning},
\begin{equation} \label{eq:grad-denoiser}
    \x_{t+1} = \x_t - \nabla_\x \reg(\x_t),
\end{equation}
which also serves as the regularization update of the reconstruction.
The network $N_\sigma$ is then trained by learning \cref{eq:grad-denoiser} as a denoiser through second-order autodifferentiation~\cite{hurault2022gradient}. The gradient has a closed form
\begin{equation} \label{eq:gs-denoiser}
    \nabla_\x \reg(\x) = \x - N_{\sigma}(\x) - \jac_{N_{\sigma}}(\x)\TT \bigl(\x - N_{\sigma}(\x)\bigr),
\end{equation}
which requires an extra backward pass for each evaluation, increasing computational complexity.
In exchange, the gradient structure by construction allows provable descent of the explicit functional $\reg$.

In recent years, diffusion- and flow-based models have demonstrated impressive results in image generation. A major approach is to learn the score of the noisy data distribution $\nabla_\x \log p_\sigma(\x)$ to allow sampling from $p(\x)$.
By Tweedie's formula, the MMSE estimators and the score are interchangeable via
\begin{equation}
    D_{\sigma}(\x) = \x + \sigma^{2}\,\nabla \log p_{\sigma}(\x).
\end{equation}
The \emph{ScorePrior} approach thus constructs the regularization as gradient ascent on $\log p_{\sigma_{t}}$:
\begin{equation} \label{eq:score}
    \x_{t+1} = \x_t + \gamma_t \frac{D_{\sigma_t}(\x_t) - \x_t}{\sigma_t^2}
\end{equation}
with step size $\gamma_t$ and noise level $\sigma_t$ both annealed to $0$ throughout the iterations, similar to the scenario of unconditional generation.

Another route keeps the regularization term explicit and constrains its structure, so that the reconstruction algorithm provably converges.
Convex ridge regularizers learn a sum of one-dimensional convex potentials applied to filtered images~\cite{goujon2023neural}.
Their weakly convex extension relaxes the potentials to weakly convex ones, which keeps the convergence of proximal gradient descent to a critical point and sharpens the prior~\cite{goujon2024learning}.
Multivariate fields of experts (MFoE) replace the one-dimensional potentials by multivariate ones built from Moreau envelopes of the $\ell_2$ norm~\cite{ducotterd2025multivariate}.
The term $\reg_{\sigma}$ depends on the noise level $\sigma$ and is trained for Gaussian denoising.
The regularization step is one gradient step
\begin{equation} \label{eq:mfoe}
    \x_{t+1} = \x_t - \frac{1}{L} \nabla_\x \reg_{\sigma}(\x_t),
\end{equation}
with $L$ the Lipschitz constant of $\nabla \reg_{\sigma}$, so that the iteration is a proximal gradient descent on an explicit objective.
The MFoE term has fewer than $2 \times 10^{4}$ parameters.
In the two-layer SRPR model, it recovers the signal from fewer measurements than DRUNet~\cite{ducotterd2025breaking}.

Without regularization, the baselines are largely gradient- or projection-based iterations on the data term, such as Wirtinger flow~\cite{candes2015phase} and its truncated variants~\cite{chen2015solving,wang2018solving}.

\subsection{Related Works}
\emph{Constraints in Fourier phase retrieval.} A support constraint removes the translation ambiguity~\cite{fienup1982phase,luke2004relaxed}, an asymmetric support also removes the twin image~\cite{fienup1987reconstruction}, and value constraints such as realness~\cite{shoushtari2023diffusion,li2024difffpr,zhang2025improving} or nonnegativity~\cite{fienup1982phase,luke2004relaxed,tian2023phase} restrict the search space further, while the global phase is fixed by anchoring one pixel~\cite{zhang2024whats}.
Removing the ambiguities does not guarantee convergence. The strongest result to date shows that a continuous hybrid input--output algorithm with nonnegativity almost always converges to the true signal or its twin~\cite{tian2023phase}.
Reliable Fourier phase retrieval thus keeps relying on strong constraints that complex-valued objects do not satisfy.

\emph{Image priors for phase retrieval.} Learned priors have been applied to phase retrieval through plug-and-play~\cite{song2022fast} and diffusion models~\cite{shoushtari2023diffusion,li2024difffpr,zhang2025improving}.
These works either use multiple CDPs or assume a real-valued signal.
Furthermore, they are evaluated on simulated measurements of natural images.
To our knowledge, phase retrieval from a single CDP of a complex-valued object has not been reported, nor has any of these methods been validated on experimental data.

\emph{Single-shot phase imaging with random masks.} Random-mask holography, single-shot ptychography and lensless imaging realize single-shot acquisition through a random modulation of the object wave.
SPIRaL~\cite{horisaki2016single} retrieves a complex-valued object from one exposure under randomized illumination with a hand-crafted sparsity prior, and the recent setup of~\cite{gao2026singleshot} is closest to our one-layer model. Both mix the signal through a Fresnel propagation rather than a Fourier transform.

\section{Methods}

\subsection{Optical Setup} \label{sec:experimental_setup}

\begin{figure}[!t]
  \centering
  \includegraphics[width=\columnwidth]{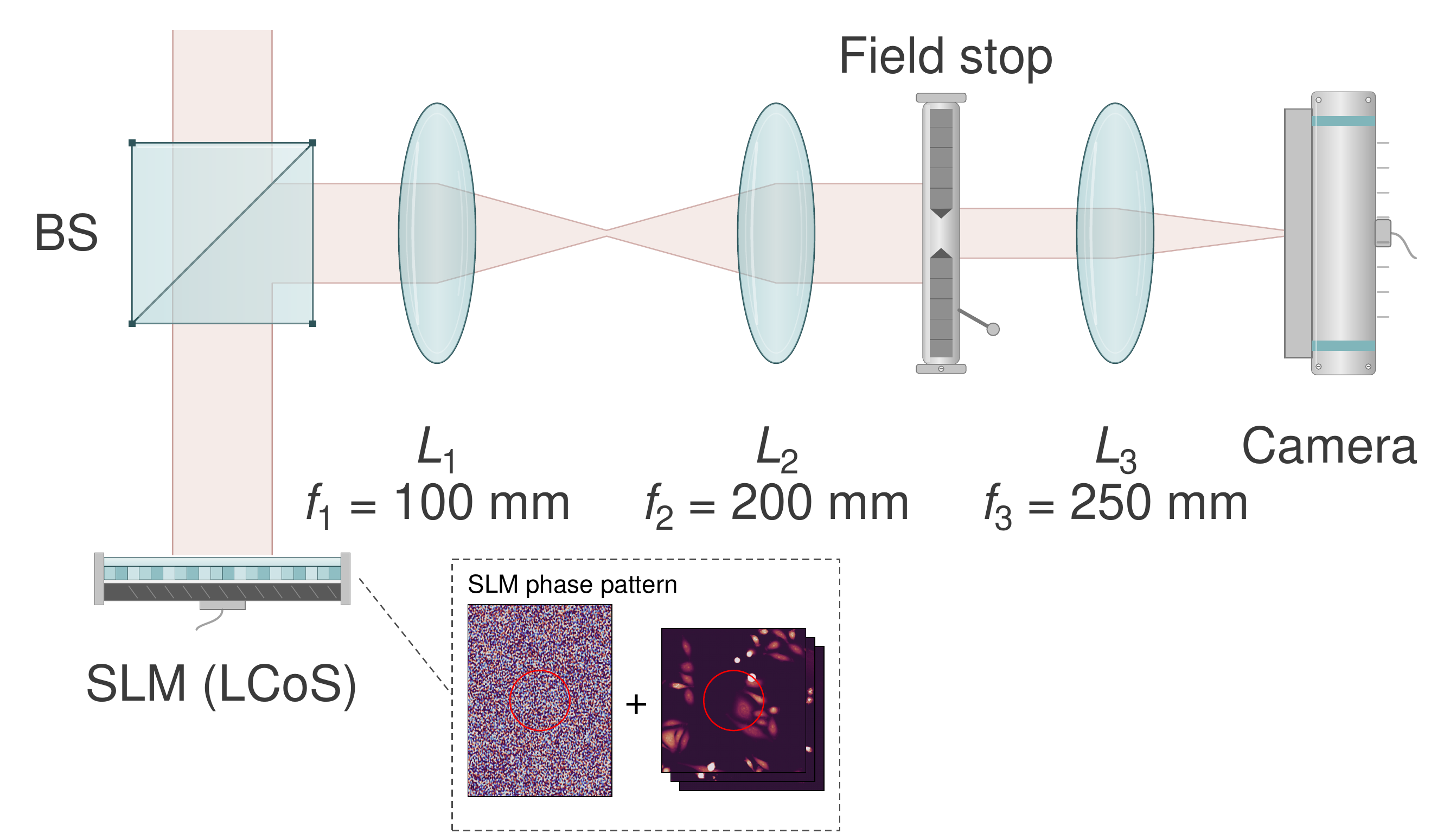}
  \caption{Optical setup implementing the 1-layer SRPR model. A collimated laser beam ($\lambda = 532$~nm) illuminates a reflective phase-only SLM (Hamamatsu, X10468-04) through a beamsplitter. The SLM plane is imaged by a $4f$ relay ($f_1 = 100$~mm, $f_2 = 200$~mm) onto an iris acting as a field stop; lens $L_3$ ($f_3 = 250$~mm) performs an optical Fourier transform and the sCMOS camera (pco.edge 4.2) in its back focal plane records the intensity. The inset shows the phase pattern displayed on the SLM, the sum of the diffuser pattern and the object phase, here frames of the Brno dataset. The red circles mark the illuminated disk of radius 125 pixels.}
  \label{fig:setup}
\end{figure}

To demonstrate reconstruction from real measurements, we built an optical system that implements the 1-layer SRPR model.
The optical setup is shown in \cref{fig:setup}.
A collimated continuous-wave laser beam at wavelength $\lambda = 532$~nm illuminates a reflective phase-only liquid-crystal spatial light modulator (SLM, Hamamatsu X10468-04, $800\times600$ pixels, pixel pitch $20~\mu$m).
The SLM displays the sum of the phase pattern of the random diffuser $\D$ and the signal of interest $\x$, the latter being added to the central $600\times600$ pixels.
The SLM plane is imaged with a $4f$ relay (focal lengths $f_1 = 100$~mm and $f_2 = 200$~mm, magnification $f_2/f_1 = 2$) onto an iris, which acts as a field stop and restricts the field to a disk of radius approximately $125$ SLM pixels centered on the SLM.
A third lens ($f_3 = 250$~mm) placed one focal length after the iris performs an optical Fourier transform, and an sCMOS camera (PCO edge 4.2, pixel size $6.5~\mu$m) in its back focal plane records the intensity of the Fourier transform of the modulated field.
Three additional captures without the object calibrate the alignment of the camera plane and the instrument phase of \cref{sec:recon-algorithm}.
They display a constant phase, a linear phase ramp and the bare diffuser $\D$.
The design of the diffuser pattern displayed on the SLM is described in \cref{sec:diffuser-design}.

For the test signals, we take 5 first signals for visual validation and the Brno QPM dataset~\cite{vicar2021annotated,vicar2021self}, which contains 530 images from 5 different cell lines obtained by a coherence-controlled holographic microscope to construct the phase $\phi_{\mathrm{obj}}$ of the object $\x = \e^{\j \phi_\mathrm{obj}}$. The pixel values are linearly remapped from $[0, 1]$ to $[0, \pi]$.

The camera captures the diffraction pattern of the combined SLM pattern on a $512 \times 512$ pixel grid in the Fraunhofer regime, and the unknowns of the reconstruction are the pixels of the illuminated disk of radius 125.
The nominal sampling ratio, i.e., the number of camera pixels over the number of unknowns, is thus $512^2 / (\pi \cdot 125^2) \approx 5.3$.
However, the measurements are distributed highly non-uniformly.
\Cref{fig:measurements} compares the calibrated forward model with a captured frame.
The intensity concentrates in a central ellipse of semi-axes $64 \times 48$ pixels, which carries 90\% of the energy above the noise floor.
A secondary halo of semi-axes $160 \times 120$ pixels surrounds and carries the remaining 10\%.
The aspect ratio $H/W = 4/3$ follows from the band-limit disk of the speckle pattern being defined on the $H \times W$ DFT grid of the SLM frame (\cref{sec:diffuser-design}).
Outside the halo, the intensity is at the noise floor of the camera and is removed by the noise subtraction of the reconstruction.
Counting only the pixels inside the halo, the effective sampling ratio is $(160 \times 120) / 125^{2} \approx 1.23$, slightly above the weak-recovery threshold of Gaussian random models.

\subsection{Designing the Diffuser Pattern} \label{sec:diffuser-design}
To generate the random diffuser pattern, we first assign independent phase values drawn uniformly from $(-\pi, \pi]$ to each superpixel of a given size in the SLM.
This introduces discrete changes at the boundaries of superpixels, which is unfavorable due to the display crosstalk of adjacent pixels in the SLM.
To enforce smoothness in the diffuser pattern, we apply the Fourier transform on band-limited random frequencies, i.e.,
\begin{equation}
    \D = \e^{\j \angle(\F^{-1}\mathbf{S}\F \e^{\j \boldsymbol{\Psi}})},
\end{equation}
where $\F$ is the 2D DFT and $\boldsymbol{\Psi} \in \R^{H \times W}$ is uniformly sampled from $(-\pi, \pi]$.
The mask $\mathbf{S}$ is a centered disk of radius $r$ samples on the DFT grid of the $H \times W$ frame, which cuts off the high frequencies in the spectrum of $\e^{\j \boldsymbol{\Psi}}$.
One sample of this grid is $1/W$ cycles per pixel along the side of $W$ pixels and $1/H$ cycles per pixel along the side of $H$ pixels, so the two cutoff frequencies $r/W$ and $r/H$ differ by the factor $H/W$.
This produces smooth speckle grains in $\D$ whose size is set by $r$.
We define the overall speckle grain size as
\begin{equation}
    g = \sqrt{g_x g_y}, \qquad g_x = \frac{W}{2r}, \quad g_y = \frac{H}{2r},
    \label{eq:grain-size}
\end{equation}
where $g_x$ and $g_y$ represent the speckle grain sizes per axis, i.e. the half periods of the two cutoff frequencies in pixels.
Larger superpixels or speckle grains alleviate the crosstalk of the SLM, but also introduce less randomness into the system.
We design three levels of speckle sizes, $r = 125$, $75$ and $47$, i.e. $g = 2.8$, $4.6$ and $7.4$ pixels, and visualize them in \cref{fig:diffuser-patterns}, together with three levels of i.i.d. patterns with superpixel sizes of 3, 5 and 8 pixels, respectively.

We compare the six patterns of \cref{fig:diffuser-patterns} with the Shepp-Logan phantom~\cite{shepp1974fourier}, remapped to phases in $[0, \pi/2]$ and captured with every pattern.
The reconstruction runs PnP-HQS (\cref{alg:unified}) with $\sigma_0 = 3.0$~rad and $\rho=0.9995$ for 1900 iterations.
\Cref{fig:pattern-choice} shows the Shepp-Logan phantom ground truth together with the reconstructions.
It can be seen that the speckle pattern with $g = 4.6$ is the only pattern where the closed outline is fully recovered.
We thus use this configuration for later captures with more signals.

\subsection{Reconstruction Algorithm} \label{sec:recon-algorithm}
The forward operator is modeled on the $512 \times 512$ camera grid as $\y = |\F \D \S \e^{\j \boldsymbol{\phi}_\mathrm{cali}} \x|^{2}$.
Here $\D$ is the known random mask, $\S$ the support of the illuminated disk and $\boldsymbol{\phi}_\mathrm{cali} \in [-\pi, \pi)^{512 \times 512}$ the instrument phase obtained from calibration.
The DFT and its adjoint, the inverse DFT, are evaluated on the same grid.
The support radius is 130 pixels instead of the 125 pixels of the iris, a margin of 5 pixels that keeps the hard support edge away from the illumination edge.

Before reconstruction, we locate the zero order on the capture of the constant phase and refine the alignment of the camera plane by cross-correlation with the simulated capture of the bare diffuser.
On the same capture, we fit the instrument phase $\boldsymbol{\phi}_\mathrm{cali} = a r^{2} + b r^{4} + g$, with $r$ the distance to the grid center and $g$ a small Gaussian bump.
The capture of the ramp serves as a check of the alignment.
Each measurement is then scaled to the norm of the simulated measurement of a flat object.

All frameworks of \Cref{sec:algorithms} run the same iteration, summarized in \Cref{alg:unified}.
The simulations in \cref{sec:results} use the PnP-HQS instance of \cref{alg:unified}.
Their test signal, initialization and schedule are stated with each experiment.
\begin{algorithm}[t]
\caption{Regularized reconstruction from a single CDP}
\label{alg:unified}
\begin{algorithmic}
\REQUIRE measurement $\y$, operator $\A = \F\D\e^{\j\boldsymbol{\phi}_\mathrm{cali}}$, support $\S$, denoiser $D_{\sigma}$, step size $\mu$, schedule $(\sigma_t, \lambda_t)_{t < T}$
\STATE $\x_0 \sim \mathcal{CN}(\mathbf{0}, \mathbf{I})$
\FOR{$t = 0, \dots, T - 1$}
    \STATE $\x' \leftarrow \x_t + \mu\, \A\HH \bigl(\sqrt{\y} \odot \e^{\j \angle \A\x_t} - \A\x_t\bigr)$ \COMMENT{data step, \cref{eq:relaxed_measurement_projection}}
    \STATE $\x'' \leftarrow \x' + \lambda_t \bigl(D_{\sigma_t}(\x') - \x'\bigr)$ \COMMENT{prior move, \cref{eq:prior-move}}
    \STATE $\x_{t+1} \leftarrow \S \x''$ \COMMENT{support projection}
\ENDFOR
\end{algorithmic}
\end{algorithm}
Each iteration applies three steps to the current iterate $\x_t$.
The data step is the relaxed measurement projection of \cref{eq:relaxed_measurement_projection} with $\A = \F\D\e^{\j\boldsymbol{\phi}_\mathrm{cali}}$ and step size $\mu$.
This $\A$ is unitary on the $512 \times 512$ grid, so the projection and proximal statements of \Cref{sec:algorithms} apply.
The prior move is
\begin{equation} \label{eq:prior-move}
    \x'' = \x' + \lambda_t \bigl(D_{\sigma_t}(\x') - \x'\bigr),
\end{equation}
with the noise level $\sigma_t = \sigma_0 \rho^{t}$ and the weight $\lambda_t = \eta\,(\sigma_t / \sigma_0)^{p}$.
The support step is the hard projection $\x_{t+1} = \S\x''$ onto the illuminated disk.
The frameworks differ in $\mu$, $\lambda_t$ and the denoiser, as listed in \Cref{tab:frameworks-params}.
Before reconstruction, a constant noise level, set per framework, is subtracted from the captured intensity, and negative values are set to zero.
\begin{table}[t]
\centering
\caption{Instances of \Cref{alg:unified}. $\mu$ is the step size of the data step and $\lambda_t$ the weight of the prior move.}
\label{tab:frameworks-params}
\begin{tabular}{lccl}\hline
Framework & $\mu$ & $\lambda_t$ & Denoiser $D_{\sigma}$ \\ \hline
PnP-HQS & 1 & 1 & DRUNet \\
PnP-PGD & 0.75 & 1 & DRUNet \\
RED & 1 & 0.97 & DRUNet \\
ScorePrior & 1 & $(\sigma_t / \sigma_0)^{1/4}$ & DRUNet \\
LSR & 1 & 1 & $\mathbf{I} - \nabla \reg$, gradient-step DRUNet \\
MFoE & 1 & 1 & $(\mathbf{I} - \frac{2}{L} \nabla \reg_{\sigma})^{30}$ \\ 
Without prior & 1 & 0 & none \\ \hline
\end{tabular}
\end{table}
With $\mu = 1$ and $\lambda_t = 1$, the iteration alternates the exact projection with the denoiser and is PnP-HQS, \cref{eq:pnp}.
With $\mu < 1$, the data step becomes a gradient step on the amplitude loss and the iteration is PnP-PGD.
With a constant $\lambda_t = \tau < 1$, the prior move is the RED step of \cref{eq:red}.
With $\lambda_t = (\sigma_t / \sigma_0)^{p}$ decaying with the noise level, the prior move is the annealed score ascent of \cref{eq:score}, where $\lambda_t$ plays the role of $\gamma_t / \sigma_t^{2}$.
With the gradient-step denoiser $D_{\sigma} = \mathbf{I} - \nabla \reg$ of \cref{eq:grad-denoiser} and $\lambda_t = 1$, the prior move is the LSR step.
The MFoE framework replaces the denoiser by gradient steps on the weakly convex regularization term $\reg_{\sigma}$ of~\cite{ducotterd2025multivariate}, which is trained for Gaussian denoising at the noise level $\sigma$ and applied here at the level $\sigma_t$ of the schedule.
One gradient step of size $1 / L$ per iteration, \cref{eq:mfoe}, is the proximal-gradient iteration of~\cite{ducotterd2025breaking}.
On the experimental data this single step never leaves the random initialization, so the prior move takes 30 gradient steps of size $2 / L$ per iteration, the most accurate setting we tested.
These 30 gradient steps take about 25 times as long as one DRUNet step in the runs of \Cref{tab:frameworks}.
With $\lambda_t = 0$, the prior is removed and the iteration reduces to error reduction with a support constraint, the baseline of the last row of \Cref{tab:frameworks}.

The denoiser is the DRUNet~\cite{zhang2021plug} with the pretrained grayscale weights distributed in the DeepInv package~\cite{tachella2025deepinverse}, and its gradient-step variant~\cite{hurault2022gradient} for LSR.
The MFoE term uses the weights published by its authors.
No training or fine-tuning is performed.
The complex iterate enters the network as two real images, its real and imaginary parts, and the output is recombined, as detailed in \Cref{sec:complex-denoising}.
On the experimental data, the network receives the noise level $\sigma_t / (2\pi)$, floored at $5 \times 10^{-3}$.

The schedule is geometric with $\sigma_0 = 2.0$~rad and $\rho = 0.9995$ over 1900 iterations, 3000 for PnP-PGD.
For MFoE, the schedule starts at $\sigma_0 = 12.6$~rad, and the real and imaginary parts are mapped by $v \mapsto v / 2 + 1/2$ instead of the map of \Cref{sec:complex-denoising}.
The iterate starts from a complex Gaussian random field.
The reported reconstruction is the checkpoint, taken every 150 iterations, with the lowest phase RMSE against the ground truth on the support of radius 130 pixels.
This oracle readout upper-bounds the accuracy of any stopping rule.
A stopping rule without ground truth is one global stop iteration, tuned on a development subset of 40 frames and applied to every frame.
For RED, the stop at iteration 1200 gives an RMSE of 0.2186~rad and a PSNR of 23.47~dB, 0.008~rad above the oracle readout.
Across three random initializations, this RMSE changes by at most 0.001~rad.
\subsection{Generalizing Real-Valued Denoisers to Complex Values} \label{sec:complex-denoising}
Most existing denoisers are trained on real-valued grayscale images, whereas the phase retrieval iterate operates on a complex field with two real degrees of freedom per pixel.
The denoiser must therefore act on a real representation of $\x$.
Two natural choices exist.
The real-imaginary splitting denoises $\x_{\real}$ and $\x_{\imag}$ separately and recombines them.
The magnitude-phase splitting denoises the modulus $\abs$ and the phase $\ang$ of $\x = \abs \odot \e^{\j \ang}$ separately.
For our objects, $\x = \e^{\j \boldsymbol{\phi}}$ with a smooth phase, the parts $\cos \boldsymbol{\phi}$ and $\sin \boldsymbol{\phi}$ are smooth images wherever the phase is, so a natural-image prior applies to them directly.
On the experimental data, both parts are mapped into $[0, 1]$ by $v \mapsto v / 8 + 1/2$ before the network, and the output is mapped back.
The factor keeps the parts of a unit-modulus field well inside the training range of the network.

If the denoiser is the proximal operator of a regularization term $\reg$ on real-valued images~\cite{gribonval2011map}, the two splittings have different meanings.
The real-imaginary splitting is exactly the proximal operator of $\reg(\x_{\real}) + \reg(\x_{\imag})$, so the plug-and-play interpretation of \Cref{sec:algorithms} carries over to the complex field without approximation (\cref{thm:separability-real-imag}).
The magnitude-phase splitting only minimizes an upper bound of the complex proximal objective, which is tight for small phase differences and requires magnitudes in $[0, 1]$ (\cref{thm:mag-phase-upper-bound}).
A denoiser applied to the phase image also sees every $2\pi$ wrap as an edge.
For these reasons, all results use the real-imaginary splitting.
Both statements and their proofs are given in \cref{sec:splitting-theorems}.


\subsection{Evaluation Metrics}
All metrics are computed on a region $\Omega$, which is the entire image in simulation.
In the experiments, $\Omega$ is the central disk of radius 125 pixels, the aperture of the iris.
The reconstruction itself uses a support of radius 130~px, an additional 5~px margin that keeps the hard support edge off the illumination edge.
We first compute the cosine similarity (CS) between the reconstructed phase $\hat{\ang}$ and the ground truth phase $\ang$ as 
\begin{equation} \label{eq:cs}
    \mathrm{CS} = \frac{\left| \langle \e^{\j \hat{\ang}}, \e^{\j \ang} \rangle \right|}{\|\e^{\j \hat{\ang}}\| \cdot \|\e^{\j \ang}\|},
\end{equation}
where $\langle \mathbf{a}, \mathbf{b} \rangle = \mathbf{a}\HH \mathbf{b}$, and inner products and norms run over the pixels of $\Omega$.
CS returns a value in $[0, 1]$, with 1 indicating a perfect recovery and 0 an uncorrelated estimate.
CS is agnostic to the potential global phase shift in the reconstruction.

For other metrics on real-valued images which are sensitive to global phase shifts, we first compute the underlying global phase shift of the reconstruction as
\begin{equation}
    \theta^* = \angle \langle \e^{\j \hat{\ang}}, \e^{\j \ang} \rangle,
\end{equation}
and then correct the reconstructed phase as
\begin{equation}
    \hat{\ang}' = \hat{\ang} + \theta^*.
\end{equation}

We then compute the root-mean-square error (RMSE) between the corrected phase and the ground-truth phase, with each pixel difference wrapped to $(-\pi, \pi]$,
\begin{equation} \label{eq:rmse}
    \mathrm{RMSE} = \sqrt{\frac{
    \|\hat{\ang}' - \ang\|_2^2}{|\Omega|}},
\end{equation}
where $|\Omega|$ denotes the number of pixels in $\Omega$.
Lastly, we compute the peak signal-to-noise ratio (PSNR) from the RMSE,
\begin{equation} \label{eq:psnr}
    \mathrm{PSNR} = 20 \log_{10} \frac{\pi}{\mathrm{RMSE}}.
\end{equation}
The numerator is $\pi$ as the peak intensity because the displayed ground-truth phase spans $[0, \pi]$.
Over several frames or realizations, we report the mean of the per-frame values of each metric.

\section{Results} \label{sec:results}
\begin{figure}[t!]
    \centering
    \includegraphics[width=\linewidth]{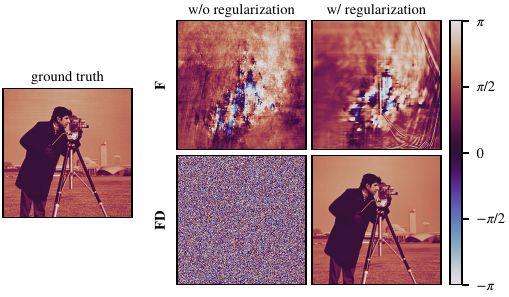}
    \caption{Effects of regularization in simulation. Ground truth phase of the cameraman (left) and its reconstructions from Fourier measurements ($\mathbf{F}$) and from a single CDP ($\mathbf{F}\mathbf{D}$) at a sampling ratio of 4, without and with regularization. Full recovery is only possible from the CDP with regularization.}
    \label{fig:effects-of-reg}
\end{figure}
\subsection{Effects of Regularization} \label{sec:effects-reg}
We study the effect of regularization on reconstruction from Fourier measurements and from a single CDP.
The ground truth is the standard cameraman test image of $256 \times 256$ pixels, used as a phase-only signal with phases in $[0, \pi]$ as for the experimental objects in \cref{sec:experimental_setup}, and shown in \Cref{fig:effects-of-reg} (left).
We zero-pad the width and height of the signal to 512 before applying $\A$, leading to a sampling ratio of 4, at which the Fourier magnitude determines almost every signal up to the trivial ambiguities~\cite{hayes1982reconstruction}.
The random diagonal $\D$ is sampled i.i.d. and uniformly from $\{-1, 1\}$.
Both reconstructions start from a spectral initialization~\cite{luo2019optimal} and run 6000 iterations of the measurement projection, \cref{eq:relaxed_measurement_projection} with $\mu = 1$, without regularization or with the PnP step of \cref{eq:pnp} before each projection.
The noise level of the DRUNet prior follows $\sigma_t = 0.2 \times 0.9997^{t}$.
In all simulations, the real and imaginary parts are mapped by $v \mapsto v / 2 + 1/2$ and clipped to $[0, 1]$ before and after the denoiser, which receives the noise level $\sigma_t$ itself.
Each denoiser call applies one of the eight flips and rotations of the image, cycled over the iterations, and inverts it afterwards. 

\Cref{fig:effects-of-reg} visualizes the phases of the reconstruction results under different
measurements and algorithms. First, we observe that the reconstruction from Fourier measurements without regularization suffers from a strong artifact resembling the twin image of the ground truth. Applying regularization slightly smooths out the twin image artifact, although the reconstruction still deviates greatly from the ground truth. In contrast, reconstruction from the CDP without regularization results in pure noise. However, applying regularization achieves perfect recovery. The potential of regularization on CDPs is thus substantial: it improves the reconstruction from a random guess to a perfect recovery, even from measurements in a single shot. The artifacts explain this difference.
Without regularization, the CDP reconstruction is a noise-like field, which a denoiser removes.
The Fourier reconstruction instead carries the twin image, a reflected copy of the object that looks as natural as the object, so a denoiser does not remove it. 

\subsection{Choice of Regularization} \label{sec:choice-reg}
We replace the DRUNet of \Cref{sec:effects-reg} by five priors that each represent one family of image regularization.
Total variation (TV) is a hand-crafted convex regularization term with a known proximal operator~\cite{rudin1992nonlinear}.
BM3D is a non-local patch-based denoiser~\cite{dabov2007image}.
EPLL is a Gaussian mixture prior on image patches~\cite{zoran2011learning}.
DnCNN is a convolutional denoiser trained for one noise level~\cite{zhang2017beyond}.
We use the public models for 15/255, 25/255 and 50/255 and apply the one trained closest to $\sigma_t$.
DRUNet takes the noise level as an input~\cite{zhang2021plug}.
MFoE is the weakly convex regularization term of \cref{sec:recon-algorithm}, applied through one gradient step per iteration~\cite{ducotterd2025multivariate}.
One gradient step moves the iterate much less than a denoiser does, and with the start $\sigma_0 = 0.2$ of the other priors, MFoE never leaves the spectral initialization.
We start its level at $\sigma_0 = 3$, since smaller values fail for some or all realizations of the forward model, and larger ones lower the accuracy.
Without noise and at sampling ratio 4 as in \Cref{fig:effects-of-reg}, all six priors reach between 35 and 45 dB, and reconstructions are indistinguishable from ground truth.
We add Gaussian noise to the measurements with a standard deviation of 0.1 times the mean intensity, which corresponds to a signal-to-noise ratio of 23 dB.
Negative values are clipped to zero.
At sampling ratio 4, every prior still exceeds 33 dB, and the reconstructions differ only slightly.
To separate the priors, we also lower the sampling ratio to 0.5, so the prior has to supply the discarded half of the coefficients.
The algorithm, the spectral initialization and the schedule $\sigma_t = 0.2 \times 0.9997^{t}$ over 6000 iterations are those of \Cref{sec:sampling-ratio}.
\Cref{fig:reg-comparison} shows the reconstructions.
All six priors converge from the same spectral start at 4.8 dB to a plateau, so the choice of prior does not decide whether the reconstruction succeeds.
It decides the accuracy and the residual artifacts.
TV flattens the face into patches at 26.5 dB.
BM3D and EPLL keep the edges but leave patch-scale residues at 29.5 dB.
DnCNN adds a fine grain at 28.5 dB.
MFoE reaches 29.5 dB with a slightly rougher texture.
Without noise it is the most accurate of the six, by 0.4 to 5 dB over DRUNet across the sampling ratios, but the gap closes at this noise level.
DRUNet gives the smoothest reconstruction at 29.9 dB, and we use it in the rest of the paper.
\begin{figure}[t!]
    \centering
    \includegraphics[width=\linewidth]{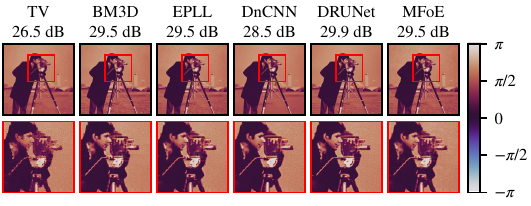}
    \caption{Reconstructions of the cameraman phase from one noisy CDP at sampling ratio 0.5 with six priors and the same algorithm and start. Gaussian noise with a standard deviation of 0.1 times the mean intensity is added to the measurements and negative values are clipped to zero. PSNR above each panel. The bottom row magnifies the $96 \times 96$ pixel region marked by the red box.}
    \label{fig:reg-comparison}
\end{figure}

\subsection{Effect of Sampling Ratios} \label{sec:sampling-ratio}
\begin{figure}[t!]
    \centering
    \begin{subfigure}[b]{\linewidth}
        \centering
        \includegraphics[width=\linewidth]{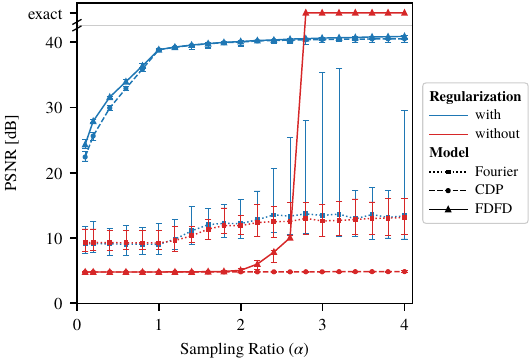}
        \caption{PSNR}
        \label{fig:sr-sweep-psnr}
    \end{subfigure}
    \begin{subfigure}[b]{\linewidth}
        \centering
        \includegraphics[width=\linewidth]{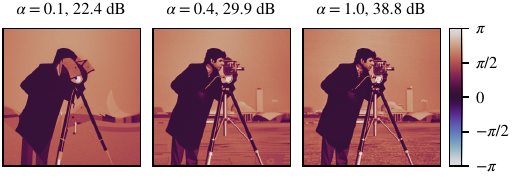}
        \caption{Reconstructed phase}
        \label{fig:sr-sweep-recon}
    \end{subfigure}
    \caption{Reconstruction accuracy under different sampling ratios. (a) Mean PSNR for Fourier measurements, a single CDP, and FDFD with two structured random layers. Without regularization, FDFD recovers the phase exactly from a sampling ratio of 2.8, while CDP stays at the level of a random phase. With regularization, both CDP and FDFD allow reconstruction from very few measurements. Fourier fails with and without regularization, except for a few regularized realizations at high sampling ratios. (b) Reconstructed phase from a CDP with regularization at sampling ratios 0.1, 0.4 and 1, for the realization whose PSNR is the closest to the mean.}
    \label{fig:sr-sweep}
\end{figure}
We investigate reconstruction quality from a single CDP with limited measurements by benchmarking the PSNR under different sampling ratios.
The same signal configuration is used as in \Cref{sec:effects-reg}.
For reference, we run the same sweep for Fourier measurements without diffuser ($L = 0.5$) and for FDFD, the two-layer SRPR model ($L = 2$) that matches dense random matrices~\cite{hu2026structured}.
The sampling ratios range from $0.1$ to $4.0$ with a step of $0.1$, while the 2D measurement size is rounded to the closest integer for each ratio.
We apply the same PnP algorithm as in \Cref{sec:effects-reg} with the noise schedule $\sigma_t = 0.2 \times 0.9997^{t}$ and 6000 iterations to allow thorough convergence.
Without regularization, the spectral initialization is followed by 6000 measurement projections.
\Cref{fig:sr-sweep} plots the average PSNR over 50 realizations of the forward model, with error bars denoting the minimum and the maximum.
Without regularization, the CDP reconstruction fails at every sampling ratio and stays at 4.8~dB, the PSNR of a uniformly random phase error, matching our previous finding in~\cite{hu2026structured}.
FDFD without regularization recovers the phase exactly from a sampling ratio of 2.8.
Fourier measurements fail with and without regularization, with a mean PSNR between 8.9 and 13.8~dB, and only 12 of the 2000 regularized runs exceed 20~dB.
With regularization, the PSNR of CDP rises monotonically from 22~dB at a sampling ratio of 0.1 to 39~dB at 1, and FDFD stays within 2.3~dB of it in \Cref{fig:sr-sweep-psnr}.
\Cref{fig:sr-sweep-recon} shows the corresponding reconstructions for one typical realization per ratio.
At a sampling ratio of 0.1, the result reconstructs the general layout of the signal with a PSNR around 22~dB.
Still, it is filled with piecewise-constant patches due to the smoothing effect of regularization under substantially insufficient measurements.
In comparison, the optimal constant-phase image yields a PSNR of 12.3~dB. 
At 0.4, the scene is recovered with a slight loss of fine texture in particular regions, while the PSNR reaches 30~dB. 
At 1, the reconstruction is visually indistinguishable from the ground truth in \Cref{fig:effects-of-reg} (left).
A fair reconstruction is thus possible well below the weak-recovery threshold of 1 of the dense random model~\cite{mondelli2019fundamental}.
This aligns with our observation in~\cite{ducotterd2025breaking}, which establishes similar results for the two-layer SRPR model.
Notably, the PSNR curve changes its slope at a sampling ratio of 1.
Below 1, the trimming operation $\mathbf{S}$ discards a fraction $1 - \alpha$ of the coefficients of $\F\D\x$.
Only the prior can supply these coefficients, so the error stops decreasing before the noise level reaches its final value.
Above 1, every coefficient is measured, and the error follows the annealing schedule down to its final noise level of 0.033, which gives the plateau at 40.5~dB.
The plateau is thus set by the schedule rather than by the measurements.
A schedule that ends at a lower noise level would raise it.
This remarkable measurement efficiency suggests strong robustness of reconstruction from experimental data.

\subsection{Experimental Results on Test Signals}
\begin{figure}[t!]
    \centering
    \includegraphics[width=\linewidth]{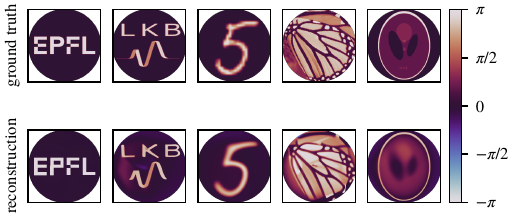}
    \caption{Phases of the five test signals and their reconstructions respectively.}
    \label{fig:experiment-test}
\end{figure}
We construct 5 test signals including two logos, one digit from MNIST~\cite{lecun1998gradient}, one butterfly and one Shepp-Logan test image. The phase configuration of the signals is detailed in \Cref{sec:experimental_setup}. The algorithm uses PnP with DRUNet and real-imaginary splitting. The noise schedule is $\sigma_t = 3.0 \times 0.995^t$, and we observe the optimization trajectory to select the best reconstruction per image. \Cref{fig:experiment-test} visualizes the phases of the ground truth signals together with their reconstructions. It can be seen that the features in the two logos and the MNIST digit can be faithfully reconstructed. For the butterfly image, the speckles on the edge of the iris cannot be reconstructed reliably, possibly due to the low illumination strength of the beam at the edge. For the Shepp-Logan image, the reconstruction slightly blurs the features and loses the fine dots in the ground truth. Nonetheless, the reconstruction still reflects most features in the original signal. 

\subsection{Experimental Results on the Brno Dataset}
\begin{table}[t]
\centering
\caption{Reconstruction frameworks on the 528 valid experimental frames of the Brno dataset. The last row is the same iteration without the denoising step.}
\label{tab:frameworks}
\begin{tabular}{lccc}
\hline
Framework & CS & RMSE [rad] & PSNR [dB] \\
\hline
RED & \textbf{0.9770} & \textbf{0.2106} & \textbf{23.92} \\
PnP-HQS & 0.9768 & 0.2110 & 23.91 \\
ScorePrior & 0.9762 & 0.2147 & 23.73 \\
PnP-PGD & 0.9758 & 0.2159 & 23.75 \\
LSR & 0.9754 & 0.2196 & 23.41 \\
MFoE & 0.9694 & 0.2350 & 23.02 \\
\hline
Without prior & 0.0368 & 1.7726 & 4.97 \\
\hline
\end{tabular}
\end{table}

We further tested reconstruction on the Brno dataset, consisting of 530 images of 5 different cell lines.
We directly apply the images of $600\times600$ pixels to the SLM pattern and reconstruct the central pixels inside the iris of a disk of radius $125$.
We observe 2 corrupted measurements in the experimental data and exclude them from the benchmark.
Apart from these two, the reconstruction converges robustly for all other signals with the five DRUNet frameworks.
\Cref{tab:frameworks} reports the average CS, RMSE and PSNR scores obtained by each framework described in \Cref{sec:algorithms}.
We observe that all five DRUNet frameworks achieve reliable reconstruction for the whole dataset, with a phase RMSE below 0.43~rad on every frame.
The five frameworks built on DRUNet lie within 0.51 dB of each other in PSNR.
Among them, the RED framework provides the best accuracy while PnP-HQS performs very closely.
MFoE, the only framework without a denoiser network, reaches 23.0 dB, 0.9 dB below RED, with four frames above 1 rad of phase RMSE, at about 25 times the cost per iteration.
In comparison, the unregularized baseline obtains only unmeaningful guesses.

\subsection{Different Oversampling Schemes}
\begin{figure}[t!]
    \centering
    \includegraphics[width=\linewidth]{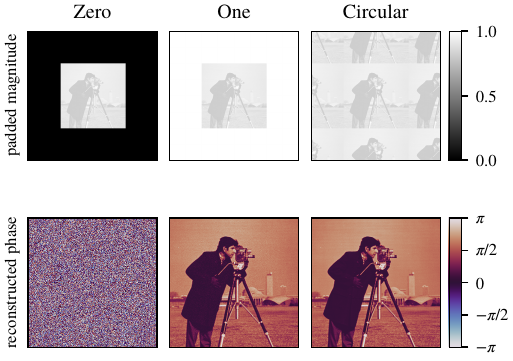}
    \caption{Effects of different padding schemes on the reconstruction from a single CDP without regularization. Both nonzero constant and circular padding do not require regularization.}
    \label{fig:padding}
\end{figure}
All results above use the zero padding as the scheme to oversample the input signal. We simulate two alternatives: circular padding and constant padding by one, and perform unregularized reconstruction on each. The signal has the phase $\pi\mathbf{c}$ and the magnitude $0.8 + 0.2\,\mathbf{c}$, where $\mathbf{c} \in [0, 1]^{256 \times 256}$ is the cameraman image.
The weak magnitude modulation makes the padding schemes visible in \cref{fig:padding}. The rest of the model and algorithm configuration is kept the same as in \Cref{sec:effects-reg}. \Cref{fig:padding} illustrates the visualization of the magnitudes of the oversampled signals and their corresponding phase reconstructions under different padding schemes. Both circular padding and padding by one enable full reconstruction of the input signal even without regularization, unlike the zero-padding case, which returns a random guess. These two padding schemes have a direct physical meaning. First, the circular padding scheme is equivalent to multi-shot captures. Indeed, the model with multiple CDPs in \cref{eq:coded-diff} can be reformulated as
\begin{equation}
    \A = \begin{bmatrix} \F\D_{1} \\ \vdots \\ \F\D_{k} \end{bmatrix} = \begin{bmatrix} \F & & \\ & \ddots & \\ && \F \end{bmatrix} \begin{bmatrix} \D_1 & & \\ & \ddots & \\ && \D_k \end{bmatrix} \begin{bmatrix} \mathbf{I} \\ \vdots \\ \mathbf{I} \end{bmatrix} = \tilde{\F} \tilde{\D} \O,
\end{equation}
where $\tilde{\F}$ is a block Fourier matrix, $\tilde{\D}$ is the diagonal matrix by aggregating $\D_i$, and $\O$ is a stack of identity matrices, which represents circular padding. The sole difference between the multi-shot model and the single-shot model with circular padding is the structured transform $\tilde{\F}$ and $\F$, which does not affect the performance of the model~\cite{hu2026structured}. Furthermore, constant padding by one is equivalent to adding a background reference field to the signal, as in holography~\cite{gabor1948new,guizar2007holography}. The results in \cref{fig:padding} show that both multi-shot and adding references are valid approaches to circumvent the necessity of regularization during reconstruction. Nonetheless, the single-shot reference-free system as modeled by zero-padding requires regularization to work.

\section{Discussion}
Our results show that an image prior turns the one-layer SRPR model, which is exactly phase retrieval from a single CDP, into a practical phase imaging modality.
Compared with the two-layer model of~\cite{hu2026structured}, the optical system reduces to a single random mask and a single Fourier transform, and one exposure of the object suffices instead of the several masks required by the existing guarantees and experimental realizations of coded diffraction imaging~\cite{seaberg2015coherent,katkovnik2017computational,wu2019wish}.

The role of the prior in this setting differs from the incremental gain in reconstruction quality that regularization usually brings.
Without a prior, reconstruction from a single CDP returns a random guess regardless of the sampling ratio, both in simulation (\Cref{fig:sr-sweep}) and on the experimental data (\Cref{tab:frameworks}). With a prior, the same measurements yield an accurate reconstruction at a sampling ratio of 1 and a fair one down to 0.1.
The prior is therefore a feasibility condition rather than a refinement.
The padding analysis in \Cref{fig:padding} offers a complementary view: the two schemes that make the measurements solvable without a prior, circular padding and constant padding, correspond physically to additional exposures and to a reference beam, that is, to additional measurements.
Regularization supplies prior information of a natural image to navigate through the ill-posed optimization landscape with a single reference-free exposure.

The prior itself requires no training.
All reconstructions use DRUNet with its publicly available weights, trained for Gaussian denoising of natural grayscale images, and the real-imaginary splitting of \Cref{thm:separability-real-imag} applies it to the complex-valued field without any modification.
Other denoisers and the explicit MFoE term reach the same recovery in simulation, which indicates that the essential ingredient is the generic structure of natural images captured by the prior rather than a particular architecture.
On the experimental data, the MFoE term needs several gradient steps per iteration to leave the random initialization.
With 30 steps, it stays $0.4$ to $0.9$~dB below the DRUNet frameworks, which shows that the implicit prior of a denoiser is also the more economical one.

The five DRUNet frameworks compared in \Cref{tab:frameworks} share one algorithmic skeleton and differ in the weight of the denoising step, in the data step and in the noise subtraction.
On the experimental dataset they perform within $0.51$~dB of each other, and the explicit MFoE term is $0.4$ to $0.9$~dB below them.
RED reaches the best accuracy, but its margin over PnP-HQS ($4 \times 10^{-4}$~rad in RMSE) is comparable to the variation between random initializations.
The choice of framework thus matters much less than the presence of the prior.

The main limitation of our experiment is that the same SLM displays both the object phase and the random mask.
This provides a reliable experimental reference for the metrics, but it leaves aside complexities of real-world samples with a diffuser at another physical plane.

A natural extension of this work is to apply the method to biological samples.
A second direction is uncertainty quantification over the random diffuser.
The accuracy varies across diffuser realizations, as the error bars in \Cref{fig:sr-sweep-psnr} show.
The pixel-wise spread of reconstructions from several diffusers would give a confidence map for the reconstructed phase.
Together, they would make the method attractive for time-resolved observation of living cells.
\appendices
\section{Diffuser Patterns and Camera Measurements} \label{sec:app-diffuser} \label{sec:camera-measurements}
\Cref{fig:diffuser-patterns} shows the six diffuser patterns of \cref{sec:diffuser-design} on the full SLM frame.
The independent patterns are piecewise constant, and the magnified insets show the phase jumps at every superpixel boundary.
The speckle patterns vary smoothly without such jumps, and their grains grow as the radius $r$ of the spectral disk decreases.
\begin{figure}[t!]
    \centering
    \includegraphics[width=1.0\linewidth]{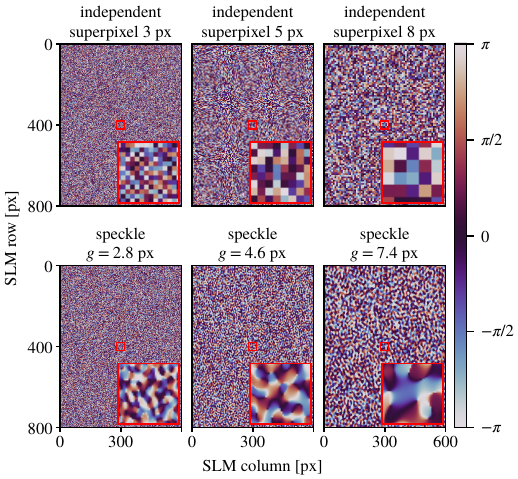}
    \caption{Diffuser patterns at three levels of randomness on the full $800 \times 600$ SLM frame. The red box marks the $40 \times 40$ pixel region at the SLM center that is magnified in the lower-right corner of each panel. Top: independent phases per superpixel of 3, 5 and 8 pixels. Bottom: speckle patterns with disk radii $r = 125$, $75$ and $47$ samples, i.e. grain sizes $g = 2.8$, $4.6$ and $7.4$ pixels according to \cref{eq:grain-size}.}
    \label{fig:diffuser-patterns}
\end{figure}

\Cref{fig:pattern-choice} shows the reconstructions of the same Shepp-Logan phantom with the six patterns.
All six reconstructions show the outer ring and the two inner ellipses of the phantom, with different completeness.
The independent patterns lose parts of the outer ring, and the patterns with 3 and 8 pixels add a bright artifact near the center of the disk.
The speckle patterns recover more of the ring.
With $g = 2.8$ and $7.4$ pixels, the ring stays open at the lower left.
Only the speckle pattern with $g = 4.6$ pixels recovers the outer ring as a closed outline, together with the bright region at the top of the head.
\begin{figure}[t!]
    \centering
    \begin{subfigure}[c]{0.33\linewidth}
        \centering
        \includegraphics[width=\linewidth]{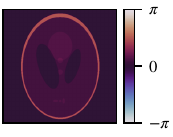}
        \caption{Ground Truth}
        \label{fig:pattern-choice-gt}
    \end{subfigure}
    \hfill
    \begin{subfigure}[c]{0.65\linewidth}
        \centering
        \includegraphics[width=\linewidth]{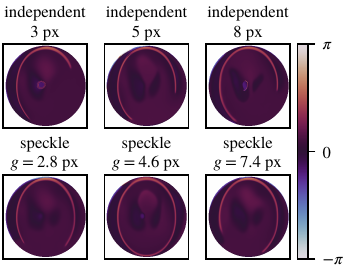}
        \caption{Reconstruction}
        \label{fig:pattern-choice-recon}
    \end{subfigure}
    \caption{Shepp-Logan phantom of the earlier batch reconstructed with the six diffuser patterns, in the order of \cref{fig:diffuser-patterns}. Pixels outside the support of radius 120 pixels, 5 pixels beyond the beam radius of this batch, are blank.}
    \label{fig:pattern-choice}
\end{figure}

\Cref{fig:measurements} compares the simulated and the captured measurement of one frame.
Both frames concentrate their intensity in a central plateau and a halo, marked by the two ellipses on the captured frame.
Beyond the halo, the simulated intensity stays three to four orders of magnitude below the maximum up to the edge of the grid.
The captured intensity is at the noise floor of the camera there and vanishes after noise subtraction.
The captured frame also has a bright spot at its center, which the forward model does not contain.
Before reconstruction, the pixels within 6 pixels of the center are replaced by the median of the surrounding ring of width 5 pixels.
\begin{figure}[t!]
    \centering
    \includegraphics[width=1.0\linewidth]{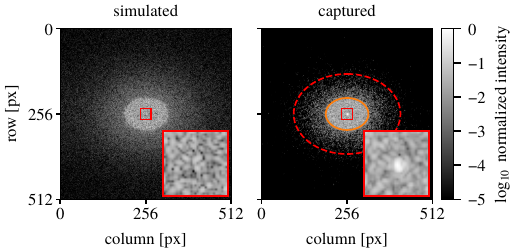}
    \caption{Simulated and captured measurement of the same frame in logarithmic scale. The captured frame is shown after the noise subtraction applied before reconstruction. The solid ellipse (semi-axes $64 \times 48$ pixels) encloses the central plateau and the dashed ellipse (semi-axes $160 \times 120$ pixels) the halo. The inset in the lower-right corner of each panel magnifies the central $32 \times 32$ pixels marked by the red box, on the same intensity scale.}
    \label{fig:measurements}
\end{figure}

\section{Proximal Operators of the Two Complex-Valued Splittings} \label{sec:splitting-theorems}
We assume that the denoiser is the proximal operator of a regularization term $\reg$ on real-valued images~\cite{gribonval2011map}.
Under this assumption, the two splittings have different meanings.
\Cref{thm:separability-real-imag} shows that the real-imaginary splitting is exactly the proximal operator of $\reg(\x_{\real}) + \reg(\x_{\imag})$, so the plug-and-play interpretation of \Cref{sec:algorithms} carries over to the complex field without approximation.
\begin{theorem} \label{thm:separability-real-imag}
    Let $\reg$ be a regularization term on real-valued images with proximal operator $\prox_{\lambda \reg}$, $\lambda > 0$.
    For every complex image $\x = \x_{\real} + \j \x_{\imag}$,
    \begin{equation}
       \prox_{\lambda (\reg(\x_{\real}) + \reg(\x_{{\imag}}))}(\x) = \prox_{\lambda \reg}(\x_{\real}) + \j \prox_{\lambda \reg}(\x_{{\imag}}).
   \end{equation}
\end{theorem}
\begin{proof}
    The squared distance separates into the real and imaginary parts, so the minimization does as well,
    \begin{align}
        &\prox_{\lambda (\reg(\x_{\real}) + \reg(\x_{{\imag}}))}(\x) \\
        = &\argmin_{\mathbf{z}}\frac{1}{2}\| \mathbf{z}- \x \|^{2}+ \lambda (\reg(\mathbf{z}_{\real}) + \reg(\mathbf{z}_{{\imag}})) \\
        = &\argmin_{\mathbf{z}_{\real} + \j \mathbf{z}_{{\imag}}}\frac{1}{2}\| \mathbf{z}_{\real}- \x_{\real} \|^{2}+ \lambda \reg(\mathbf{z}_{\real})+ \\
        & \frac{1}{2}\| \mathbf{z}_{{\imag}}- \x_{{\imag}} \|^{2}+ \lambda \reg(\mathbf{z}_{{\imag}}) \\
        = &\prox_{\lambda \reg}(\x_{\real}) + \j \prox_{\lambda \reg}(\x_{{\imag}}),
   \end{align}
   which is the claim.
\end{proof}
The induced regularization term is not invariant to a global phase rotation, which mixes the two parts.
In practice the iteration settles on one global phase, and the metrics are aligned to it.

\Cref{thm:mag-phase-upper-bound} shows that the magnitude-phase splitting minimizes an upper bound of the complex proximal objective.
It is therefore a majorization-minimization step rather than an exact proximal step.
\begin{theorem} \label{thm:mag-phase-upper-bound}
    Let $\abs, \hat{\abs} \in [0, 1]^{n}$ be magnitudes and $\ang, \hat{\ang} \in \R^{n}$ phases.
    For every regularization term $\reg$ and $\lambda > 0$,
    \begin{gather}
        \frac{1}{2}\| \hat{\abs} \e^{\j \hat{\ang}} - \abs \e^{\j \ang} \|^{2}+ \lambda \reg(\hat{\abs}) + \lambda \reg(\hat{\ang}) \leq \\ \frac{1}{2}\| \hat{\abs} - \abs \|^{2}+ \lambda \reg(\hat{\abs}) + \frac{1}{2}\| \hat{\ang} - \ang \|^{2}+ \lambda \reg(\hat{\ang}).
    \end{gather}
\end{theorem}
\begin{proof}
    Expanding the squared norm and using $1 - \cos \beta \leq \beta^{2} / 2$ and $r_{i} \hat{r}_{i} \leq 1$,
    \begin{align*}
        & \frac{1}{2}\| \hat{\abs} \e^{\j \hat{\ang}} - \abs \e^{\j \ang} \|^{2}+ \lambda \reg(\hat{\abs}) + \lambda \reg(\hat{\ang}) \\
        = & \frac{1}{2} \| \hat{\abs} \|^{2} + \frac{1}{2} \| \abs \|^{2} - \sum_{i=1}^{n} r_{i} \hat{r}_{i} \cos{(\hat{\theta}_{i} - \theta_{i})} + \lambda \reg(\hat{\abs}) + \lambda \reg (\hat{\ang}) \\
        = & \frac{1}{2} \| \hat{\abs}  - \abs \|^{2} + \sum_{i=1}^{n} r_{i} \hat{r}_{i} (1 - \cos{(\hat{\theta}_{i} - \theta_{i})}) + \lambda \reg(\hat{\abs}) + \lambda \reg (\hat{\ang}) \\
        \leq & \frac{1}{2} \| \hat{\abs}  - \abs \|^{2} + \sum_{i=1}^{n} r_{i} \hat{r}_{i} \frac{{(\hat{\theta}_{i} - \theta_{i})}^{2}}{2} + \lambda \reg(\hat{\abs}) + \lambda \reg (\hat{\ang}) \\
        \leq & \frac{1}{2} \| \hat{\abs}  - \abs \|^{2} + \lambda \reg(\hat{\abs}) + \frac{1}{2}\| \hat{\ang} - \ang \|^{2}+ \lambda \reg (\hat{\ang}),
    \end{align*}
    which is the claim.
\end{proof}
The bound is tight for small phase differences, uses the unwrapped phase difference and requires magnitudes in $[0, 1]$.

\bibliography{refs}
\bibliographystyle{IEEEtran}
\end{document}